\documentclass[a4paper,UKenglish,cleveref,autoref,thm-restate]{lipics-v2021}

\usepackage{amsmath}
\usepackage{amsthm}
\usepackage{thmtools}
\usepackage{amssymb}
\usepackage{bbm}
\usepackage{mathtools}
\usepackage{tikz}
\usepackage{MnSymbol}
\usepackage{lipsum,hyperref}
\usepackage{graphicx}
\usepackage{amsfonts}
\usepackage{xfrac}
\usepackage{color,soul}
\usepackage{framed}
\usepackage{setspace}
\usepackage{dutchcal}
\usepackage{rotating}
\usepackage{setspace}
\usepackage{subcaption}
\usepackage{xstring}
\usepackage{apxproof}
\usepackage{footnote}
\usepackage{multirow}
\usepackage[framemethod=TikZ]{mdframed}

\nolinenumbers

\graphicspath{{./drawings/}}

\hypersetup{colorlinks=true}	

\newcommand{\eqdef}{\stackrel{\text{def}}{=}}

\newtheorem*{corollary*}{Corollary}
\newtheorem*{definition*}{Definition}

\DeclareMathOperator{\mon}{mon}

\newcommand{\norm}[1]{\left\lVert#1\right\rVert}

\newcommand{\adeg}[1]{\widetilde{\deg}\left( #1 \right)}

\newcommand{\epsadeg}[1]{\widetilde{\deg}_{\epsilon}\left( #1 \right)}

\DeclareMathOperator{\Per}{per}

\DeclareMathOperator{\rank}{rank}

\DeclareMathOperator{\fs}{fs}

\DeclareMathOperator{\id}{id}

\DeclareMathOperator{\GL}{GL}

\DeclareMathOperator{\DCC}{D^{CC}}

\DeclareMathOperator{\PM}{PM}

\DeclareMathOperator{\MCn}{MC_n}

\DeclareMathOperator{\ANDDT}{AND-DT}
\DeclareMathOperator{\XORDT}{XOR-DT}

\DeclareMathOperator{\D}{D}
\DeclareMathOperator{\DXOR}{\D^{\XOR}}
\DeclareMathOperator{\DAND}{\D^{\AND}}

\DeclareMathOperator{\BPM}{BPM}
\DeclareMathOperator{\BPMn}{BPM_n}
\DeclareMathOperator{\UBPM}{UBPM}
\DeclareMathOperator{\UBPMn}{UBPM_n}
\DeclareMathOperator{\BPMnstar}{BPM_n^\star}
\DeclareMathOperator{\UBPMnstar}{UBPM_n^\star}
\DeclareMathOperator{\BMnk}{BM_{n,k}}
\DeclareMathOperator{\UBMnk}{UBM_{n,k}}
\DeclareMathOperator{\BMnkstar}{BM_{n,k}^\star}
\DeclareMathOperator{\BM}{BM}
\DeclareMathOperator{\MaxMatchnk}{MaxMatch_{n,k}}

\DeclareMathOperator{\XOR}{XOR}
\DeclareMathOperator{\AND}{AND}

\DeclareMathOperator{\ANDn}{AND_n}

\DeclareMathOperator{\ORn}{OR_n}

\newcommand\restr[2]{{\left.\kern-\nulldelimiterspace#1\vphantom{\big|}\right|_{#2}}}

\definecolor{nicegrey}{HTML}{e6e6ea}
{\endMakeFramed}

{\endMakeFramed}

\definecolor{niceorange}{HTML}{fed766}
\newenvironment{theo}[2][]{%
	\refstepcounter{theo}
	\ifstrempty{#1}%
	{\mdfsetup{%
			frametitle={%
				\tikz[baseline=(current bounding box.east),outer sep=0pt]
				\node[anchor=east,rectangle,fill=blue!20]
				{\strut Theorem~\thetheo};}
		}%
	}{\mdfsetup{%
			frametitle={%
				\tikz[baseline=(current bounding box.east),outer sep=0pt, inner ysep=1pt]
				\node[anchor=east,rectangle, draw=nicegrey, line width=0.7pt, fill=white]
				{\strut Theorem~\thetheo:~#1};}%
		}%
	}%
	\mdfsetup{%
		innertopmargin=1pt,linecolor=nicegrey,%
		linewidth=0.7pt,topline=true,%
		frametitleaboveskip=\dimexpr-\ht\strutbox\relax%
	}
	\begin{mdframed}[]\relax}{%
\end{mdframed}}

\definecolor{niceblue}{HTML}{0e9aa7}

\newcommand{\keepvalues}{%
	\edef\restorevalues{%
		\parindent=\the\parindent
		\parskip=\the\parskip
	}%
}

\newcounter{theo}[section]
\renewcommand{\thetheo}{\arabic{theo}}

\title{Algebraic Representations of Unique Bipartite Perfect Matching}
\titlerunning{Algebraic Representations of Unique Bipartite Perfect Matching}

\author{Gal Beniamini}{The Hebrew University of Jerusalem}{gal.beniamini@mail.huji.ac.il}{}{}
\authorrunning{G. Beniamini}
\Copyright{Gal Beniamini}

\ccsdesc[500]{Mathematics of computing~Matchings and factors}
\ccsdesc[500]{Theory of computation~Communication complexity}
\ccsdesc[500]{Theory of computation~Oracles and decision trees}

\keywords{Bipartite Perfect Matching, Boolean Functions, Partially Ordered Sets}

\category{}
\relatedversion{}
\EventEditors{}
\EventNoEds{2}
\EventLongTitle{}
\EventShortTitle{}
\EventAcronym{}
\EventYear{2022}
\EventDate{}
\EventLocation{}
\EventLogo{}
\SeriesVolume{}
\ArticleNo{1}

\begin{document}

\maketitle

\begin{abstract}
	We obtain complete characterizations of the Unique Bipartite Perfect Matching function, and of its Boolean dual, using multilinear polynomials over the reals. Building on previous results \cite{beniamini2020approximate,beniamini2020bipartite}, we show that, surprisingly, the dual description is \textit{sparse} and has \textit{low $\ell_1$-norm} -- only exponential in $\Theta(n \log n)$, and this result extends even to other families of matching-related functions. Our approach relies on the M\"obius numbers in the matching-covered lattice, and a key ingredient in our proof is M\"obius' inversion formula.
	
	These polynomial representations yield complexity-theoretic results. For instance, we show that unique bipartite matching is \textit{evasive} for classical decision trees, and \textit{nearly evasive} even for generalized query models. We also obtain a tight $\Theta(n \log n)$ bound on the log-rank of the associated two-party communication task. 
\end{abstract}


\section{Introduction}

A perfect matching in a graph is a subset of edges spanning the graph, no two of which are incident to the same vertex. In this paper we consider the \textit{decision problem} of unique bipartite matching: the input is a balanced bipartite graph over $2n$ vertices, and the goal is to determine whether the graph contains a \textit{unique} perfect matching. This problem can be naturally cast as a Boolean function. 

\begin{definition*}
	The unique bipartite perfect matching function $\UBPMn: \{0,1\}^{n^2} \to \{0,1\}$ is
	\[ \UBPMn(x_{1,1}, \dots, x_{n,n}) = \begin{cases} 1 & \big\{ (i,j) : x_{i,j} = 1\big\} \text{ has a unique perfect matching}  \\ 0 & \text{otherwise} \end{cases} \]
\end{definition*}



The complexity of $\UBPMn$ is closely related to that of $\BPMn$ -- the problem in which we drop the uniqueness condition and simply ask whether a bipartite graph  \textit{contains} a perfect matching. Both $\BPMn$ and $\UBPMn$ are known to lie in $\mathbf{P}$. The former due to a classical result by Edmonds \cite{edmonds1965paths}, and the latter due to Gabow, Kaplan and Tarjan \cite{gabow2001unique}. However, despite their close connection, not all known algorithmic results extend from one problem to another. For instance, $\UBPMn$ was shown by Kozen, Vazirani and Vazirani to be in $\mathbf{NC}$ \cite{kozen1985nc} (see also \cite{hoang2006bipartite}), and no such result is known for $\BPMn$. Lov\'asz showed that $\BPMn$ is in $\mathbf{RNC}$ \cite{lovasz1979determinants}, and the current best-known \textit{deterministic} parallel algorithm is due to Fenner, Gurjar and Thierauf \cite{fenner2019bipartite}, placing the problem in $\mathbf{Quasi}$-$\mathbf{NC}$. Determining the membership of bipartite perfect matching in $\mathbf{NC}$ remains one of the main open problems in parallelizability.

Our main results in this paper are the complete characterizations of both $\UBPMn$ and its dual function, by means of polynomials. These characterizations leverage a deep connection to the polynomial representations of $\BPMn$, obtained in \cite{beniamini2020bipartite, beniamini2020approximate}, and it is our hope that they can be used to further our understanding of the connection between the two. To present our results we require some notation. We say that a bipartite graph is \textit{matching-covered} if every edge of the graph participates in some perfect matching. For a graph $G$ we denote its \textit{cyclomatic number}, a topological quantity, by $\chi(G) = e(G) - v(G) + c(G)$. The set of all perfect matchings of $G$ is denoted $\PM(G)$, and the cardinality of this set is denoted $\Per(G)$ (the permanent of $G$). Under these notations, our first theorem is the following closed-form description of the \textit{unique} real multilinear polynomial representing $\UBPMn$.

\begin{theo}[The Unique Bipartite Perfect Matching Polynomial]{thm:ubpm_polynomial}
	\label{thm:ubpm_polynomial}
	$$ \UBPMn(x_{1,1}, \dots, x_{n,n}) = \sum_{G \subseteq K_{n,n}} c_G \prod_{(i,j) \in E(G)} x_{i,j} $$
	\quad where
	$$
		c_G = \begin{cases} (-1)^{\chi(G)} \Per(G) & G \text{ is matching-covered} \\
		0 & \text{otherwise} \end{cases}
	$$
\end{theo}

The polynomial appearing in Theorem \ref{thm:ubpm_polynomial} bears a striking resemblance to the representation of $\BPMn$, the only difference being the multiplicative $\Per(G)$ appearing in each term of $\UBPMn$. This is a direct result of the connection between the two functions and the \textit{matching-covered lattice}, hereafter $\mathcal{L}_n$, which is formed by all order $2n$ matching-covered graphs, ordered with respect to the subgraph relation. Billera and Sarangarajan \cite{billera1994combinatorics} proved that $\mathcal{L}_n$ is isomorphic to the face lattice of the Birkhoff Polytope $\mathbf{B}_n$. Consequently, this combinatorial lattice is Eulerian, and its M\"obius function is particularly well-behaved -- a fact which we rely on indirectly throughout this paper. In \cite{beniamini2020bipartite}, it was shown that $\BPMn$ is intimately related to the matching-covered lattice: every such graph corresponds to a monomial, and their coefficients are given by M\"obius numbers. Our proof of Theorem~\ref{thm:ubpm_polynomial} extends this connection by leveraging M\"obius Inversion Formula, and in fact allows us to derive the polynomial representation for \textit{any} indicator function over $\mathcal{L}_n$ (including, for instance, $\BPMn$), while also simplifying somewhat parts of the original proof.

Theorem \ref{thm:ubpm_polynomial} yields information-theoretic lower bounds. For example, $\UBPMn$ has full total degree and is thus \textit{evasive}, i.e., any decision tree computing it must have full depth, $n^2$. Unlike its analogue $\BPMn$, which is a \textit{monotone} bipartite graph property and thus known to be evasive \cite{kahn1984topological}, the \textit{unique} perfect matching function is \textit{not monotone}, and for such functions evasiveness is not guaranteed (see e.g. \cite{lovasz2002lecture}). We also obtain lower bounds against generalized families of decision trees, whose internal nodes are labeled by arbitrary parity functions ($\XORDT$), or conjunctions ($\ANDDT$), over subsets of the inputs bits. 

\begin{corollary*}
		For classical, parity, and conjunction trees, the following lower bounds hold:
		$$\D(\UBPMn) = n^2,\ \ \DXOR(\UBPMn) \ge \left(\tfrac12 - o(1)\right) n^2 \ \ and\ \ \DAND(\UBPMn) \ge (\log_3 2) n^2 - o(1)$$
\end{corollary*}

In the second part of this paper we consider the Boolean dual function $\UBPMnstar$, which is obtained by flipping all the input and output bits (or formally, $\UBPMnstar(x_{1,1}, \dots, x_{n,n}) = 1-\UBPMn(1-x_{1,1}, \dots, 1-x_{n,n})$). By construction, this is the indicator over all bipartite graphs whose \textit{complement} does \textit{not} contain a unique perfect matching. Our second result is a complete characterization of $\UBPMnstar$ as a real multilinear polynomial. This description relies \textit{heavily} on the that of  $\BPMnstar$ -- which is the dual of the bipartite perfect matching function $\BPMn$. The polynomial representation of the latter dual was obtained in a series of papers \cite{beniamini2020bipartite, beniamini2020approximate}, and is omitted here for brevity. 

\begin{theo}[The Dual Polynomial of Unique Bipartite Perfect Matching]{thm:ubpm_dual_polynomial}
	\label{thm:ubpm_dual_polynomial}
	$$ \UBPMnstar(x_{1,1}, \dots, x_{n,n}) = \sum_{G \subseteq K_{n,n}} c^\star_G \prod_{(i,j) \in E(G)} x_{i,j} $$
	\quad where
	$$
	c^\star_G = \Per(G) \cdot a_G^\star + \sum_{M \notin \PM(G)} (-1)^{|E(M) \setminus E(G)|} \cdot a_{G \cup M}^\star
	$$
	\quad and $a_{G}^\star$ denotes the coefficient of $G$ in $\BPMnstar$.
\end{theo}

Theorem \ref{thm:ubpm_dual_polynomial} expresses the coefficient of every graph $G$ as an alternating sum over coefficients of $\BPMnstar$, corresponding exactly to those graphs formed by adjoining a single perfect matching to $G$. This suffices in order to \textit{inherit} the main structural result of \cite{beniamini2020approximate} regarding $\BPMnstar$: the $\ell_1$-norm of $\UBPMnstar$, i.e., the norm of the coefficient vector of the representing polynomial, is \textit{very small} -- only exponential in $\Theta(n \log n)$, and this is tight.

\begin{corollary*}
	The dual polynomial is \textit{sparse} and its coefficients have \textit{small}. Explicitly,
	$$ \log \norm{\UBPMnstar}_1 = \Theta(n \log n)$$
\end{corollary*}

The low norm of the dual yields algorithmic results for the unique-bipartite-matching problem, and for related matching problems. For instance, through the approximation scheme of \cite{beniamini2020approximate, sherstov2020algorithmic}, it allows one to obtain a low-degree polynomial \textit{approximation} of the unique bipartite matching function over the hypercube (i.e., ``approximate degree''), which holds even for \textit{exponentially small} error. The same $\ell_1$-norm bound also directly extends to the spectral norm of $\UBPMn$,\footnote{The spectra of any function and its dual are identical up to sign, and the $\{0,1\}$-polynomial $\ell_1$-norm is always trivially at least as large as the $\{\pm1\}$-representation (``Fourier'') $\ell_1$-norm.} which is a well-studied quantity in analysis of Boolean functions.



Finally, we consider the two-party deterministic communication complexity of unique bipartite matching. The input is a graph $G \subseteq K_{n,n}$, whose edges are distributed among two parties according to \textit{any} \textit{arbitrary} and \textit{fixed} partition. The sparse polynomial representation of $\UBPMnstar$ allows us to deduce that the log-rank of the communication matrix, for \textit{any} of the above communication tasks, is bounded by only $\mathcal{O}(n \log n)$, and we prove that this is tight.\footnote{In fact, our results hold even for a certain \textit{$\land$-lifted} and dualised version of this problem.} We remark that, while we show that unique matching has low log-rank, not much is known regarding its \textit{deterministic communication complexity}. For the monotone variant $\BPMn$, known algorithms (e.g. \cite{hopcroft1973n}) can be translated into protocols using only $\widetilde{\mathcal{O}}(n^{\sfrac32})$ bits \cite{nisan2021demand}. However, it is currently not known how to convert algorithms for $\UBPMn$ (such as \cite{gabow2001unique, golumbic2001uniquely}), into protocols using even $\mathcal{O}(n^{2-\varepsilon})$ bits, for any $\varepsilon > 0$. Determining the deterministic communication complexity of $\UBPMn$ is thus left as an open problem.

\section{Preliminaries and Notation}

\subsection{Boolean Functions and Polynomials}

Every Boolean function $f: \{0,1\}^n \to \{0,1\}$ can be \textit{uniquely} represented by a multilinear polynomial $p \in \mathbb{R}\left[x_1, \dots, x_n \right]$ (see e.g. \cite{o2014analysis}), where $f$ and $p$ agree on all Boolean inputs $\{0,1\}^n$. The family of subsets corresponding to monomials in this polynomial representation (i.e., whose coefficient does not vanish) is denoted by $\mon(f)$. The cardinality of $\mon(f)$ is known as the \textit{sparsity} of $f$, and the maximal cardinality of any $S \in \mon(f)$ is known as the \textit{total degree} of $f$, hereafter $\deg(f)$. The $\ell_1$-norm of $f$ is the norm of its representing polynomial's coefficient vector, namely:
$$ \norm{f}_1 \eqdef \norm{ (a_S)_{S \subseteq [n]} }_1 \text{, where $f$ is $\{0,1\}$-represented by } p(x_1, \dots, x_n) = \sum_{S \subseteq [n]} a_S \prod_{i \in S} x_i $$

Given a Boolean function $f: \{0,1\}^n \to \{0,1\}$, it is often useful to consider the transformation in which we \textit{invert} all the input and output bits. This process produces a new Boolean function $f^\star$, known as the Boolean dual. 

\begin{definition}
	\label{defn:boolean_dual}
	Let $f: \{0,1\}^n \rightarrow \{0,1\}$ be a Boolean function. The \textbf{Boolean dual} of $f$ is the function $f^\star: \{0,1\}^n \to \{0,1\}$ where the symbols $0$ are $1$ are interchanged. Formally,
	$$ \forall x \in \{0,1\}^n:\ f^\star(x_1, \dots, x_n) = 1 - f(1-x_1, \dots, 1-x_n) $$
\end{definition}

The polynomial representations of a Boolean function $f$ and its dual $f^\star$ can differ substantially (for example $\ANDn^\star = \ORn$, and while the former is represented by a single monomial, the latter consists of $2^n-1$ monomials). However, since $f$ and $f^\star$ are obtained by affine transformations of one another, they share many properties. For example, their Fourier spectra are identical \cite{o2014analysis}, up to sign. Moreover, they have the same approximate degree \cite{beniamini2020approximate} for any error $\varepsilon$, and the ranks of their associated communication matrices (see proceeding subsections) are identical up to an additive constant (of $1$).

\subsection{Graphs}

We use the standard notation for quantities relating to graphs. 
A less common measure appearing in this paper is the cyclomatic number $\chi(G)$, a topological quantity.
\begin{definition}
	\label{def:cyclomatic_number}
	Let $G$ be a graph. The \textbf{cyclomatic number} of $G$ is defined by: 
	$$ \chi(G) = e(G) - v(G) + c(G) $$
\end{definition}

A \textit{matching} in a graph $G \subseteq K_{n,n}$ is a collection of edges sharing no vertices, and said matchings are called \textit{perfect} if they contain exactly $n$ edges (i.e., every vertex in the graph is incident to precisely one edge in the matching). The set of \textit{all} perfect matchings denoted by $\PM(G)$. For any graph $G \subseteq K_{n,n}$, we define the \textit{permanent} $\Per(G)$ and the \textit{determinant} $\det(G)$ as the application of these two functions to the biadjacency matrix of $G$, noting that $\Per(G)$ counts the number of perfect matchings in $G$.

Perfect matchings and the graphs formed by unions thereof play a central role in this paper. A graph $G \subseteq K_{n,n}$ is called \textbf{matching-covered} if and only if every edge of $G$ participates in some perfect matching.
%
%
Matching-covered graphs have interesting combinatorial properties. For example, this is precisely the family of all graphs admitting a bipartite ear decomposition (similar to the ear decomposition of 2-edge-connected graphs). This family had previously appeared extensively in the literature, and in particular had been studied at length by Lov{\'a}sz and Plummer \cite{plummer1986matching}, and by Hetyei \cite{hetyei1964rectangular}. Hereafter, we denote the set of all such graphs by 
$$\MCn = \Big\{ G \subseteq K_{n,n} : G \text{ is matching-covered} \Big\} $$

All graphs in this paper are \textit{balanced bipartite graphs}, over the fixed vertex set of the complete bipartite graph $K_{n,n}$. Consequently, we use the notation $G \subseteq H$ to indicate inclusion over the edges, and similarly $G \cup H$ is the graph whose edges are $E(G) \cup E(H)$. Lastly, many of the Boolean functions appearing in this paper are defined over subgraphs of $K_{n,n}$, where every input bit is associated with a single edge. For such functions, the notation $f(G)$, where $G \subseteq K_{n,n}$, corresponds to this mapping.

\subsection{Communication Complexity}

In this paper we consider the \textit{two-party deterministic communication model}. For a comprehensive textbook on the topic, we refer the reader to \cite{communication_complexity_book}. 
The \textbf{deterministic communication complexity} of $f$, hereafter $\DCC(f)$, is the \textit{least number of bits communicated by a protocol computing $f$, on the worst-case input}. Any (\textit{unpartitioned}) Boolean function $f: \{0,1\}^n \to \{0,1\}$ naturally gives rise to a \textit{family} of associated two-party communication tasks: one corresponding to each possible partition of the input bits between the two parties. The \textbf{deterministic communication complexity of a Boolean function} is then defined as follows.

\begin{definition}
	The deterministic communication complexity of $f: \{0,1\}^n \to \{0,1\}$ is defined
	$$ \DCC(f) \eqdef \max_{S \sqcup \bar{S} = [n]} \DCC \big(f_S(x,y)\big) $$
	where $\DCC\big(f_S(x,y)\big)$ is the deterministic communication complexity of the two-party Boolean function $f_S(x,y): \{0,1\}^{|S|} \times \{0,1\}^{|\bar{S}|} \to \{0,1\}$, representing $f$ under the partition $S \sqcup \bar{S}$.
\end{definition}

For any two-party Boolean function, let us also define the following two useful objects.

\begin{definition}
	Let $f: \{0,1\}^m \times \{0,1\}^n \to \{0,1\}$. The \textbf{communication matrix} of $f$ is	$$
	M_f \in \mathbb{R}^{\{0,1\}^m \times \{0,1\}^n},\text{  where }\forall (x,y) \in \{0,1\}^m \times \{0,1\}^n:\ M_f(x,y) = f(x,y)
	$$
\end{definition}

\begin{definition}
	\label{def:fooling_set}
	Let $f: \{0,1\}^m \times \{0,1\}^n \to \{0,1\}$ be a two-party function. We say that $S \subseteq \{0,1\}^m \times \{0,1\}^n$ is a \textbf{fooling set} for $f$ if and only if:
	$$S \subseteq f^{-1}(1), \text{ and } \forall (x_1,y_1)\ne(x_2,y_2) \in S:\ \{(x_1,y_2), (x_2,y_1)\} \cap f^{-1}(0) \ne \emptyset $$
\end{definition}

The log of the rank of $M_f$ \textit{over the reals} (sometimes referred to as the ``log-rank of $f$'') is intimately related to the communication complexity of $f$. A classical theorem due to Mehlhorn and Schmidt \cite{mehlhorn1982vegas} states that $\DCC(f) \ge \log_2 \rank M_f$, and these two quantities are famously conjectured to be polynomially related \cite{lovasz1988lattices}. As for the fooling set, it is well known that $\DCC(f) \ge \log_2 \fs(f)$ for any two-party function $f$ \cite{communication_complexity_book}, where $\fs(f)$ is the maximum size of a fooling set. This bound was extended by Dietzfelbinger, Hromkovi{\v{c}} and Schnitger \cite{dietzfelbinger1996comparison}, who showed that in fact $\log \fs (f) \le 2 \log \rank f + 2$.

\subsection{Posets, Lattices and M\"obius Functions}

Partially ordered sets (hereafter, \textbf{posets}) are defined by a tuple $\mathcal{P} = (P,\le)$, where $P$ is the element set, and $\le$ is the order relation (which is reflexive, antisymmetric and transitive). For any two elements $x,y \in P$, the notation $[x,y] \eqdef \{ z \in P : x \le z \le y \}$ denotes the \textit{interval} from $x$ to $y$. 
A \textit{combinatorial lattice} is a poset satisfying two additional conditions: every two elements have a least upper bound (a ``join''), and a greatest lower bound (a ``meet''). The \textit{face lattice of a polytope} is a combinatorial lattice whose elements correspond to the faces of a polytope, ordered by the subset relation. Such a lattice is \textit{bounded} -- it has a unique bottom element (the empty face $\hat{0}$), and a unique top element (the polytope itself), and it is also \textit{graded}, meaning that the length of all maximal chains between any two elements $x,y$ are identical (in other words, the elements can be \textit{ranked}).

Partially ordered sets  come equipped with an important function known as the \textbf{M\"obius function}. The M\"obius function of a poset is the inverse, with respect to convolution, of its zeta function $\zeta(x,y) = \mathbbm{1}\{x < y\}$.  For information on incidence algebra and the M\"obius function, we refer the reader to \cite{Stanley:2011:ECV:2124415}.

\begin{definition}[M\"obius Function for Posets]
	\label{mobius_func_poset}
	Let $\mathcal{P} = (P, \le)$ be a finite poset. The M\"obius function $\mu_\mathcal{P} : P \times P \rightarrow \mathbb{R}$ of $\mathcal{P}$ is defined
	$$ \forall x \in P:\ \mu_\mathcal{P}(x, x) = 1, \quad \forall x,y \in P,\ y < x:\ \mu_\mathcal{P}(y,x) = - \sum_{y \le z < x} \mu_\mathcal{P}(y,z)$$
\end{definition}

The M\"obius Inversion Formula allows one to relate two functions defined on a poset $\mathcal{P}$, where one function is a downwards closed sum of another, by means of the M\"obius function. This can be seen as a generalization of its number-theoretic analogue (as indeed the M\"obius function of number theory arises in this manner from the \textit{divisibility poset}).  

\begin{theorem}[M\"obius Inversion Formula, see \cite{Stanley:2011:ECV:2124415}]
	\label{thm:mobius_inversion_formula}
	Let $\mathcal{P} = (P, \le)$ be a finite poset and let $f,h: P \rightarrow \mathbb{F}$ be two functions, where $\mathbb{F}$ is a field. Then:
	$$\forall x \in P:\ h(x) = \sum_{y \le x} f(y) \quad\iff\quad \forall x \in P:\ f(x) = \sum_{y \le x} h(y)\mu_\mathcal{P}(y,x)$$	
\end{theorem}

\section{The Unique Perfect Matching Polynomial}

Our main object of study is the unique bipartite matching function.
\begin{definition}
	The Unique Bipartite Perfect Matching function is defined
	$$ \UBPMn(x_{1,1}, \dots, x_{n,n}) = \begin{cases} 1 & \big\{ (i,j) : x_{i,j} = 1\big\} \subseteq K_{n,n} \text{ has a unique P.M.}  \\ 0 & \text{otherwise} \end{cases} $$
\end{definition}
The unique multilinear representation of $\UBPMn$ is characterized in the following Theorem.
\begin{theorem}
	\label{thm:ubpm_poly}
	The unique polynomial $\UBPMn: \{0,1\}^{n^2} \rightarrow \{0,1\}$ is given by
	$$ \UBPMn\left(x_{1,1}, \dots, x_{n,n}\right) = \sum_{G \in \MCn} (-1)^{\chi(G)} \ \Per(G) \prod_{(i,j) \in E(G)} x_{i,j}$$
\end{theorem}
\begin{proof}
	The proof is centered around the combinatorial \textit{lattice of matching-covered graphs},
	$$ \mathcal{L}_n = \left(\MCn \cup \{\hat{0}\}, \subseteq\right), \text{  where $\hat{0}$ is the graph with $2n$ isolated vertices} $$
	where the order relation for this lattice is \textit{containment over the edge set}, i.e., $G \supseteq H \iff E(G) \supseteq E(H)$. Let us consider the following two functions $f: \mathcal{L}_n \to \{0,1\}$ and $h: \mathcal{L}_n \to \mathbb{Z}$ on the lattice, which are the restrictions of $\UBPMn$ and of the Permanent function, respectively:
	\begin{equation*}
		\begin{split}
			\forall G \in (\MCn \cup \{\hat{0}\}):\ f(G) &= \UBPMn(G) = \begin{cases} 
			1 & G \in PM(K_{n,n}) \\
			0 & \text{otherwise}
			\end{cases} \\
			h(G) &= \Per(G) = \#\text{Perfect Matchings in $G$}
		\end{split}
	\end{equation*}
	These two functions are intimately related. Indeed, for any element $G$ of the lattice, one can compute $h(G)$ by taking the sum $f(H)$ over all $H$ in the downwards closed interval $[\hat{0}, G]$. Therefore, by an application of M\"obius' Inversion Formula (Theorem \ref{thm:mobius_inversion_formula}) to the matching-covered lattice, we obtain:
	$$ \forall G \in \mathcal{L}_n:\ h(G) = \sum_{G \supseteq H \in \mathcal{L}_n} f(H) \ \iff\ \forall G \in \mathcal{L}_n:\ f(G) = \sum_{G \supseteq H \in \mathcal{L}_n} \mu(H,G) h(H) $$
	where $\mu: \mathcal{L}_n \rightarrow \mathbb{Z}$ is the M\"obius function of the lattice $\mathcal{L}_n$. A well known result due to Billera and Sarangarajan \cite{billera1994combinatorics} states that $\mathcal{L}_n$ is isomorphic to the \textit{face lattice} of the Birkhoff Polytope $B_n$, which is the convex hull of all $n \times n$ permutation matrices. Consequently, $\mathcal{L}_n$ is an Eulerian lattice -- and its M\"obius function $\mu$ is can be directly computed (see e.g. \cite{Stanley:2011:ECV:2124415}), as follows:
	$$ \forall G,H \in \mathcal{L}_n,\ H \subseteq G:\ \ \mu(H,G) = (-1)^{\rank(G)-\rank(H)}$$
	where $\rank(x)$ denotes the maximal length of a chain from $\hat{0}$ to $x$ (equivalently, $\rank(x) = dim(f_x)+1$, where $f_x$ is the face of $B_n$ corresponding to the lattice element $x$). In \cite{beniamini2020bipartite} it was shown that the rank of every graph $G$ in the matching-covered lattice is exactly $\chi(G)+1$, where $\chi(G) = e(G) - v(G) + c(G)$ is the cyclomatic number, a topological quantity. Recalling our prior application of M\"obius inversion, we obtain the following set of identities (note that the bottom element can be omitted, as $\Per(\hat{0})$ is zero):
	$$
		\forall G \in \mathcal{L}_n:\ (-1)^{\chi(G)}\sum_{G \supseteq H \in \MCn} (-1)^{\chi(H)} \Per(H) = \begin{cases} 
			1 & G \in \PM(K_{n,n}) \\
			0 & \text{otherwise} 
		\end{cases}
	 $$
	 To conclude the proof, let us consider the following real multilinear polynomial, wherein we assign weight $(-1)^{\chi{G}} \Per(G)$ to every matching-covered graph: 
	 $$
	 	p(x_{1,1}, \dots, x_{n,n}) =  \sum_{G \in \MCn} (-1)^{\chi(G)} \Per(G) \prod_{(i,j) \in E(G)} x_{i,j}
	 $$
	 Let $G \subseteq K_{n,n}$ and observe that, by construction:
	 $$ p(G) = \sum_{H \in \MCn} (-1)^{\chi(H)} \Per(H) \cdot \mathbbm{1}\big\{E(H) \subseteq E(G) \big\}= \sum_{G \supseteq H \in \MCn} (-1)^{\chi(H)} \Per(H) $$
	 It remains to show that $p$ ``agrees'' with $\UBPMn$ on all inputs. It is not hard to see that it suffices to show this claim only for matching-covered graphs, since given any $G \subseteq K_{n,n}$ which is \textit{not} matching-covered, one may consider the graph $G'$ formed by the union of all perfect matching in $G$ (in other words, the maximal matching-covered graph contained in $G$). By construction, we have $p(G') = p(G)$, and by definition, $\UBPMn(G) = \UBPMn(G')$ -- thus, hereafter we consider only inputs $G \in \MCn$. First, let us check the two trivial cases; the empty graph, and a single matching:
	 $$ p(\hat{0}) = 0 \text{, and } p(M) = (-1)^{\chi(M)} = (-1)^{n - 2n + n} = 1 \ \ \ \ \forall M \in PM(K_{n,n}) $$
	 Finally, for any matching-covered graph $G$ containing \textit{more than a single matching}, i.e. $G \in \MCn$ such that $G \notin PM(K_{n,n})$, it holds that:
	 $$ p(G) = \sum_{G \supseteq H \in \MCn} (-1)^{\chi(H)} \Per(H) = 0 $$
	 where the last equality follows from the identities obtained through M\"obius' Inversion Formula. Thus, $p(x_{1,1}, \dots, x_{n,n})$ agrees with $\UBPMn$ everywhere, and is its \textit{unique} representation. 
\end{proof}

\subsection{Indicators on the Matching-Covered Lattice}

We remark that the proof of Theorem \ref{thm:ubpm_poly} readily extends, through the same analysis using M\"obius inversion, to any arbitrary \textit{indicator function} over the matching-covered lattice. For any set $S \subseteq \MCn$, let $I_S: \{0,1\}^{n^2} \to \{0,1\}$ be the Boolean function
$$
	\forall G \subseteq K_{n,n}:\ I_S(G) = \mathbbm{1}\left\{ H \in S \text{ where } H=\bigcup_{M \in \PM(G)}M \right\}
$$
Then, the multilinear polynomial representing $I_S$ is given by
$$
	I_S(x_{1,1}, \dots, x_{n,n}) = \sum_{G \in \MCn}  \left( (-1)^{\chi(G)} \sum_{H \in [\hat{0}, G] \cap S} (-1)^{\chi(H)+1} \right) \prod_{(i,j) \in E(G)} x_{i,j}
$$

\subsection{Evasiveness and Generalized Decision Trees}

The characterization of $\UBPMn$ as a multilinear polynomial can be used to derive several complexity-theoretic corollaries. Firstly, this polynomial has \textit{full total degree over $\mathbb{R}$} and thus (see e.g. \cite{buhrman2002complexity}):
\begin{corollary}
$\UBPMn$ is \textit{evasive}, i.e., any decision computing it has full depth, $n^2$.
\end{corollary}
Let us remark that, contrary to its counterpart $\BPMn$ which is a \textit{monotone} bipartite graph property and thus known to be evasive \cite{yao1988monotone}, the \textit{unique} matching function is \textit{not monotone} and for such functions evasiveness is not guaranteed (see \cite{lovasz2002lecture} for one such example). Theorem \ref{thm:ubpm_poly} can be also used to derive strong bounds (near evasiveness) versus larger classes of decision trees, for example trees whose internal nodes are labeled by arbitrary conjunctions of the input bits (hereafter $\ANDDT$), and by arbitrary parity functions ($\XORDT$). It is known \cite{beniamini2020bipartite} that the depth of any $\ANDDT$ computing a Boolean function $f$ is at least $\log_3{|\mon(f)|}$. Applying this to $\UBPMn$ and recalling that asymptotically almost all balanced bipartite graphs are matching-covered (\cite{beniamini2020bipartite}), we have:
\begin{corollary}
Any $\ANDDT$ computing $\UBPMn$ has depth at least $(\log_3 2) \cdot n^2 - o_n(1)$. 
\end{corollary}
As for parity decision trees, it is well known that the depth of any such tree is bounded by the total degree of its unique representing polynomial, over $\mathbb{F}_2$ (see \cite{o2014analysis, beniamini2020bipartite}). Noting that $\Per(G) \equiv \det(G) \pmod 2$, we may write the $\mathbb{F}_2$-polynomial representation of $\UBPMn$ as follows
$$
	\UBPMn\left(x_{1,1}, \dots, x_{n,n}\right) = \sum_{\substack{G \in \MCn \\ \det(G) \equiv 1 \pmod 2}} \prod_{(i,j) \in E(G)} x_{i,j}
$$

Clearly this polynomial does not have full degree for any $n > 1$, as $\Per(K_{n,n})$ is $n! \equiv 0 \pmod 2$\footnote{It is well known (\cite{o2014analysis}) that for any function $f: \{0,1\}^n \rightarrow \{0,1\}$, $\deg_{2}(f) = n \iff |f^{-1}(1)| \equiv 1 \pmod 2$. Therefore we obtain that the number of graphs $G \subseteq K_{n,n}$ containing a \textit{unique perfect matching} is even, for any $n>1$.}. Nevertheless, we claim that its $\mathbb{F}_2$-degree is at most a constant factor away from full. Observe that its monomials constitute precisely of all graphs that are both matching-covered, and whose biadjacency matrices are invertible over $\mathbb{F}_2$, i.e., are elements of the group $\GL_n(\mathbb{F}_2)$. However, asymptotically almost all graphs are matching-covered, and by a standard counting argument, the order of $\GL_n(\mathbb{F}_2)$ satisfies
$$ 
	\Pr_{A \sim M_{n}(\mathbb{F}_2)}[A \in \GL_n(\mathbb{F}_2)] = \left(\tfrac12; \tfrac12\right)_\infty \pm o_n(1)
$$
where $\left(\tfrac12; \tfrac12\right)_\infty \approx 0.28878$ is a Pochhammer symbol. Thus by a standard Chernoff argument, there exists a matching-covered graph with odd determinant and at least $\tfrac{1}2 n^2 - o_n(1)$ edges.
\begin{corollary}
	$\DXOR(\UBPMn) \ge \deg_{2}(\UBPMn) \ge \left(\tfrac 12 - o_n(1)\right) n^2$
\end{corollary}
%
%
%
%

\section{The Dual Polynomial}

In this section we consider the Boolean dual function (Definition \ref{defn:boolean_dual}) of $\UBPMn$.

\begin{definition}
	The function $\UBPMnstar: \{0,1\}^{n^2} \to \{0,1\}$ is defined
	$$ \UBPMnstar(x_{1,1}, \dots, x_{n,n}) = \begin{cases} 1 & \big\{ (i,j) : x_{i,j} = 0\big\} \subseteq K_{n,n} \text{ does \underline{not} have a unique P.M.}  \\ 0 & \text{otherwise} \end{cases} $$
\end{definition}

In what follows, we provide a full characterization of polynomial representing $\UBPMnstar$. This description relies heavily on the that of another dual function -- $\BPMnstar$ -- which is the dual of the bipartite perfect matching function $\BPMn$ (which is defined identically to $\UBPMn$, but without the \textit{uniqueness} condition). 
%
%
The polynomial representation of $\BPMnstar$ was obtained in a series of papers \cite{beniamini2020bipartite, beniamini2020approximate}. Its monomials correspond to a family of graphs called ``\textit{totally ordered graphs}'', and their coefficients are can be computed through a normal-form block decomposition of the aforementioned graphs. The full details are presented in \cite{beniamini2020approximate}, and are omitted here for brevity. In what follows, it suffices for us to denote
$$
	\BPMnstar(x_{1,1}, \dots, x_{n,n}) = \sum_{G \subseteq K_{n,n}} a^\star_G \prod_{(i,j) \in E(G)} x_{i,j}
$$
Under this notation, our characterization of $\UBPMnstar$ is the following. 

\begin{theorem}
	\label{thm:ubpmnstar_poly}
	The unique polynomial representation of $\UBPMnstar: \{0,1\}^{n^2} \rightarrow \{0,1\}$ is 
	$$ \UBPMnstar(x_{1,1}, \dots, x_{n,n}) = \sum_{G \subseteq K_{n,n}} c^\star_G \prod_{(i,j) \in E(G)} x_{i,j} $$
	where for every $G \subseteq K_{n,n}$ we have:
	$$ c^\star_G = \Per(G) \cdot a_G^\star + \sum_{M \notin \PM(G)} (-1)^{|E(M) \setminus E(G)|} \cdot a_{G \cup M}^\star
	$$
\end{theorem}
\begin{proof}
	The polynomial representing $\UBPMnstar$ can be expressed using $\UBPMn$, via duality:
	$$ \UBPMnstar(x_{1,1}, \dots, x_{n,n}) = 1-\UBPMn(1-x_{1,1}, \dots, 1-x_{n,n}) $$
	Substituting the characterization of Theorem \ref{thm:ubpm_poly} and expanding, we deduce that the coefficient of every graph $G \subseteq K_{n,n}$ in $\UBPMn$ is:
	$$
		c^\star_G = (-1)^{e(G)+1} \sum_{G \subseteq H \in \MCn} (-1)^{\chi(H)} \Per(H)
	$$
	Writing $\Per(H) = \sum_{M \in \PM(K_{n,n})} \mathbbm{1}\{M \subseteq H \}$ and exchanging order of summation,
	$$ c^\star_G = (-1)^{e(G)+1} \sum_{M \in \PM(K_{n,n})} \sum_{G \subseteq H \in \MCn} (-1)^{\chi(H)} \mathbbm{1} \{M \subseteq H \} $$
	There are two possible cases in the above summation over all perfect matchings; either the matching is present in $G$, or it is not. Clearly every matching-covered graph containing $G$ also contains any matching of $G$, so in the former case we get a contribution of $(-1)^{e(G)+1}\Per(G) \cdot \sum_{G \subseteq H \in \MCn} (-1)^{\chi(H)}$. As for the latter case, observe that for every $M \notin \PM(G)$, the set of matching-covered graphs containing $G$ and $M$ is exactly all matching-covered graphs containing $G \cup M$. Finally, we recall \cite{beniamini2020bipartite} that the coefficient of any graph $G \subseteq K_{n,n}$ in $\BPMnstar$ is given by:
	$$
		a_G^\star = (-1)^{e(G)+1} \sum_{G \subseteq H \in \MCn} (-1)^{\chi(H)}
	$$
	Putting the two together and simplifying, we obtain: 
	$$
		c^\star_G = \Per(G) \cdot a_G^\star + \sum_{M \notin \PM(G)} (-1)^{|E(M) \setminus E(G)|} \cdot a_{G \cup M}^\star \qedhere
	$$
\end{proof}

\subsection{Corollary: The $\ell_1$-norm of $\UBPMnstar$}

One immediately corollary of Theorem \ref{thm:ubpmnstar_poly} is the following fact: the multilinear polynomial representing $\UBPMnstar$ has \textit{very low $\ell_1$-norm} -- i.e., it has few monomials, and the coefficient of every such monomial is not too large. A similar bound had previous been attained for $\BPMnstar$ in \cite{beniamini2020approximate}, which we heavily rely on for our proof. 
\begin{corollary}
	\label{cor:l1norm_ubpmnstar}
	The $\ell_1$-norm of $\BPMnstar$ is bounded only by $ \log_2 \norm{\UBPMnstar}_1 = \Theta(n \log n)$.
\end{corollary}
\begin{proof}
	For the upper bound, we rely heavily on Theorem \ref{thm:ubpmnstar_poly} and on the $\ell_1$-norm of $\BPMnstar$ obtained in \cite{beniamini2020approximate}. In the latter, it was shown that every coefficient in $\BPMnstar$ has magnitude at most $2^{2n}$, and thus using the characterization of Theorem \ref{thm:ubpmnstar_poly}, the coefficient of any graph $G$ satisfies
	 $$ \log_2 |c^\star_G| \le \log_2 \left(\Per(G) \cdot 2^{2n} + \left(n! - \Per(G)\right) \cdot 2^{2n}\right) \le n \log_2 n + n \log_2\left(4/e\right) + \Theta(\log n) $$
	It remains to bound the \textit{sparsity} of $\UBPMnstar$. To this end, consider the graphs whose coefficients do not vanish in $\BPMnstar$, and let us take a ``ball'' around every such graph $G \in \mon(\BPMnstar)$, as follows:
	$$ B(G) = \Big\{ H \subseteq K_{n,n} : \exists M \in \PM(K_{n,n}) \text{ such that } E(H) \cup E(M) = E(G) \Big\} $$
	From Theorem \ref{thm:ubpm_poly} it follows that for every graph $G$, the coefficient $c^\star_G$ \textit{does not vanish} only if either $G \in \mon(\BPMnstar)$ or there exists some $H \in \mon(\BPMnstar)$ such that $G \in B(H)$. However, each of the aforementioned balls is relatively small (in fact, can be bounded by $|B(G)| \le 2^{n} \cdot n!$), thus by the union bound:
	$$ |\mon(\UBPMnstar)| \le |\mon(\BPMnstar)|\left(1 + 2^n \cdot n!\right) = 2^{\Theta(n \log n)} $$ 
	where the last equality follows from the bound $\log_2 |\mon(\BPMnstar)| \le 2n \log_2 n + \mathcal{O}(n)$, obtained in \cite{beniamini2020bipartite}. This concludes the proof of the upper bound. The lower bound now follows directly from Theorem \ref{thm:ubpm_poly}, as it suffices to observe that the coefficient of the complete bipartite graph is $\pm \Per(K_{n,n}) = \pm (n!)$.
\end{proof}

\section{The Communication Rank of Unique Bipartite Matching}
\label{sec:com_rank}

\subsection{Rank and Polynomial Representation}

The log-rank of a Boolean function is very closely related to its representation as a multilinear polynomial. This relationship is made very evident in the case of certain ``lifted'' functions: given a Boolean function $f: \{0,1\}^n \to \{0,1\}$, one can define the following pair of functions $f_\land, f_\oplus: \{0,1\}^n \times \{0,1\}^n \to \{0,1\}$, where 
$$ \forall x,y \in \{0,1\}^n: f_\land(x,y) = f(x \land y), \text{ and } f_\oplus(x,y) = f(x \oplus y) $$
It is well known \cite{knop2021log, bernasconi1999spectral} that the rank of the communication matrices $M_{f_\land}$ and $M_{f_\oplus}$ is \textit{exactly characterized} by the \textit{sparsity} (i.e., number of monomials) of the polynomials representing $f$ in the $\{0,1\}$-basis and the $\{\pm 1\}$-basis (the Fourier basis), respectively. In other words,
\begin{align*}
	\quad\quad\quad\quad\rank (M_{f_\land}) &= \#\{\text{monomials in $\{0,1\}$-polynomial representing $f$}\} \\
	\rank (M_{f_\oplus}) &= \#\{\text{monomials in $\{-1,1\}$-polynomial representing $f$}\}
\end{align*}
The polynomial representation of a Boolean function $f$ over the $\{0,1\}$-basis, or that of its dual $f^\star$, can also be used to derive communication rank upper bounds for non-lifted functions. The following lemma gives such a bound for the communication task of $f$, under \textit{any} input partition.

\begin{lemma}
\label{lemma:rank_upper_bound}
	Let $f: \{0,1\}^n \rightarrow \{0,1\}$. Then, for every partition $S \sqcup \bar{S} = [n]$ we have:
	$$ \rank(M_{f^{S \sqcup \bar{S}}}) \le \min\big\{\left|\mon(f)\right|, \left|\mon(f^\star)\right| +  1\big\} $$ 
\end{lemma}
\begin{proof}
Let $S \sqcup \bar{S} = [n]$ be some input partition, and let $M$ and $M'$ be the communication matrices of $f$ and $f^\star$ under this partition, respectively. By definition of Boolean duality, we have
$M = J - M_\pi M' M_\sigma $
where $J = \mathbbm{1} \otimes \mathbbm{1}$ is the all-ones matrix, and $M_\pi$, $M_\sigma$ are the permutation matrices for
$$ \forall x \subseteq S:\ \pi(x) = S \setminus x,\ \ \forall y \subseteq \bar{S}:\ \sigma(y) = \bar{S} \setminus y $$
therefore $|\rank(M) - \rank(M')| \le 1$, and it suffices to bound the rank of $M'$. However, we now observe that the polynomial representing $f$ naturally induces a $|\mon(f)|$-rank decomposition of $M$ (and likewise $f^\star$ for $M'$), as per \cite{nisan1995rank}, by considering the following sum of rank-1 matrices: 
$$ \forall T \in \mon(f) \text{, add the rank-1 matrix } a_T \cdot \left(\mathbbm{1}_X \otimes \mathbbm{1}_Y\right) $$
where $a_T$ is the coefficient of $T$ in $f$, and
$$
X = \Big\{x : \left(T \cap S\right) \subseteq x \subseteq S \Big\},\ \ Y = \Big\{y : \left(T \cap \bar{S}\right) \subseteq y \subseteq \bar{S} \Big\} \qedhere
$$
\end{proof}

\subsection{The Rank of Unique Bipartite Matching}

The log-rank of the unique bipartite matching function, ranging over all input partitions, is exactly characterized in the following Theorem.

\begin{theorem}
	\label{thm:rank_ubpm}
	The log-rank of unique bipartite perfect matching is
	$$ \max_{E \sqcup \bar{E} = E(K_{n,n})} \log \rank(M_{\UBPM_n^{E \sqcup \bar{E}}}) = \Theta (n \log n) $$
	where $\UBPM_n^{E \sqcup \bar{E}}$ is the two-party function whose input is partitioned according to $E \sqcup \bar{E}$.
\end{theorem}
\begin{proof}
	To obtain the lower bound, we must first fix a particular input partition. Assume without loss of generality that $n = 2m$ and let us partition the left and right vertices into two sets, $L = A \sqcup B$, $R = C \sqcup D$, where $A = \{a_1, \dots, a_m\}$, $B = \{b_1, \dots, b_m\}$, $C = \{c_1, \dots, c_m\}$ and $D = \{d_1, \dots, d_m\}$.
	Hereafter we consider the input partition wherein Alice receives all the edges incident to the left vertices $A$ and Bob receives all the edges incident to the left vertices $B$. To prove our lower bound, we shall construct a \textit{fooling set} (Definition \ref{def:fooling_set}). Let us introduce some notation: for every permutation $\pi \in S_m$ and two sets $X,Y \in \{A,B,C,D\}$, the notation $\pi(X,Y)$ refers to the matching from $X$ to $Y$ using the permutation $\pi$. Formally,
	$$ \forall X,Y \in \{A,B,C,D\}: \forall \pi \in S_m:\ \pi(X,Y) \eqdef \Big\{ \{x_i, y_{\pi(i)} \} : i \in [m] \}\Big\} $$
	Under this notation, we claim that
	$$ S = \Big\{ \left(\id(A,C)\ \sqcup\ \pi(A,D),\ \ \id(B,D)\ \sqcup\ \big\{ \{b_{\pi(i)}, c_j \} : 1 \le i < j \le m \big\} \right)\ :\ \pi \in S_m \Big\} \text{ } $$
	is a fooling set for $\UBPM_n^{K_{A,R} \sqcup K_{B,R}}$, where $\id \in S_m$ is the identity element. 
	\begin{figure}[h]
		\centering
		\includegraphics[height=5.2cm]{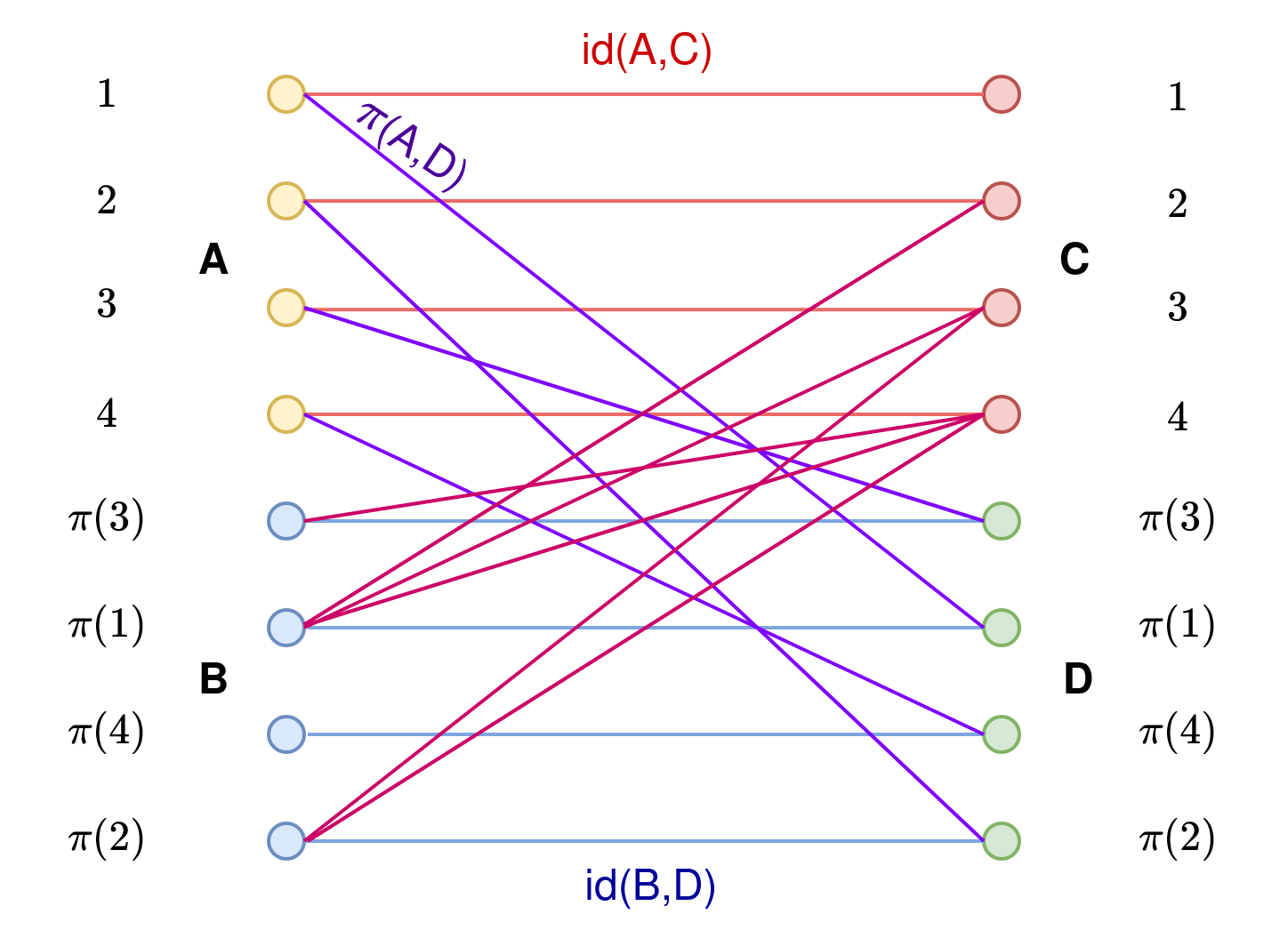}
		\caption{A graph $G$ in the fooling set $S$, for $m=4$ and $\pi=(2413)$.}
		\label{fig:fooling_set}
	\end{figure}\\
	\underline{$\{x \sqcup y : (x,y) \in S\} \subseteq \UBPM_n^{-1}(1)$:} Let $\pi \in S_m$ and consider $G \subseteq K_{n,n}$ where:
	$$ E(G) = \id(A,C)\ \sqcup\ \pi(A,D)\ \sqcup\ \id(B,D)\ \sqcup\ \big\{ \{b_{\pi(i)}, c_j \} \big\}_{i < j} $$
	Clearly $G$ has the identity perfect matching, whereby $A$ is matched to $C$ and $B$ to $D$. Let us denote this matching by $M$. To show that $M$ is \textit{unique}, it suffices to show that there exists no $M$-alternating cycle in $G$. By construction, the vertices in any such cycle must alternate between $C - A - D - B$ (since the only edges joining $A \leftrightarrow C$ and $B \leftrightarrow D$ are those in the matching $M$). Thus, for any $i \in [m]$, an $M$-alternating path starting with $c_i$ must be of the form:
	$$ c_i \sim a_i \sim d_{\pi(i)} \sim b_{\pi(i)} \sim c_{j} \sim \dots $$
	where $j > i$. However, observe that $b_{\pi(m)}$ is not adjacent to any vertex in $C$, so any such path will eventually (after at most $m$ passes through $B$) terminate at $b_{\pi(m)}$, without looping back to $c_i$. Therefore there exists no $M$-alternating cycle, and $M$ is indeed unique. \\ \ \\
	\underline{$\forall (x_1,y_1), (x_2, y_2) \in S:\ (x_1 \sqcup y_2) \in \UBPMn^{-1}(0)$:} Let $\pi, \sigma \in S_m$ where $\pi \ne \sigma$, and let $G$ be the graph:
	$$
		E(G) = \id(A,C)\ \sqcup\ \pi(A,D)\ \sqcup\ \id(B,D)\ \sqcup\ \big\{ \{b_{\sigma(i)}, c_j \} \big\}_{i < j}
	$$
	Once again, clearly $G$ has the identity matching $M$, whereby $A$ is matched to $C$ and $B$ to $D$. To show that $M$ is not unique, it suffices to exhibit an alternating cycle. Recall that $\sigma \ne \pi$ and therefore $\sigma^{-1} \circ \pi \ne \id$, and in particular, there exists some $i \in [m]$ such that $\sigma^{-1}(\pi(i)) < i$. By construction, the following $M$-alternating cycle is present in $G$:
	$ c_i \sim a_i \sim d_{\pi(i)} \sim b_{\pi(i)} = b_{\sigma\left(\sigma^{-1}\left(\pi\left(i\right)\right)\right)} \sim c_i $. \\ \ \\
	Therefore, $S$ is a fooling set for $\UBPMn$ under the aforementioned input partition. To conclude the lower bound, we recall the following Theorem, due to Dietzfelbinger, Hromkovi{\v{c}} and Schnitger \cite{dietzfelbinger1996comparison}:
	\begin{theorem}[\cite{dietzfelbinger1996comparison}]
		\label{thm:fooling_set_bounds_log_rank}
		$\forall f: \{0,1\}^{m} \times \{0,1\}^{n} \rightarrow \{0,1\}$ we have $\log_2 \fs(f) \le 2 \left( \log_2 \rank M_f + 1 \right) $.
	\end{theorem}
	\ \vspace{-0.15in} \\
	Therefore, we have:
	$$ \log_2 \rank \left(M_{\UBPM_n^{K_{A,R} \sqcup K_{B,R}}}\right) \ge \tfrac12 \log_2 |S| - 1 = \tfrac14 n \log_2 n - \Theta(n) $$
	concluding the lower bound. As for the upper bound, it follows directly from Lemma \ref{lemma:rank_upper_bound}, and from the characterization of Theorem \ref{thm:ubpmnstar_poly} (see Corollary \ref{cor:l1norm_ubpmnstar}). 
\end{proof}


%
%
%
%
%
%
%
%
%
%
%
%

\bibliography{unique_matching_polynomial}
\appendix

\section{The Approximate Degree of $\UBPMn$}
\label{appendix:apx_deg}

The $\varepsilon$-approximate degree $\epsadeg f$, of a Boolean function $f: \{0,1\}^n \to \{0,1\}$ is the \textit{least} degree of a real multilinear polynomial \textit{approximating} $f$ pointwise over $\{0,1\}^n$, with error at most $\varepsilon$. Formally,

\begin{definition}
	Let $f: \{0,1\}^n \rightarrow \{0,1\}$ and let $0 < \varepsilon < \tfrac{1}{2}$. The $\varepsilon$-approximate degree of $f$, $\epsadeg{f}$, is the least degree of a real multilinear polynomial $p \in \mathbb{R}[x_1, \dots, x_n]$ such that:
	$$ \forall x \in \{0,1\}^n:\ |f(x) - p(x)| \le \varepsilon $$ 
	If $\varepsilon = \sfrac13$, then we omit the subscript in the above notation, and instead write $\adeg f$.
\end{definition}

Approximate degree is a well-studied complexity measure. For a comprehensive survey on the topic, we refer the reader to \cite{bun2021guest}. With regards to Theorem~\ref{thm:ubpm_dual_polynomial}, we make the following observation: every Boolean function whose polynomial representation, or that of its dual, have low $\ell_1$-norm -- can be efficiently \textit{approximated} in the $\ell_\infty$-norm by a low-degree polynomial. Firstly, it is not hard to see that for any Boolean function $f: \{0,1\}^n \rightarrow \{0,1\}$ and any $\varepsilon>0$, the $\varepsilon$-approximate degree of $f$ is identical to that of its dual $f^\star$. This follows since $f^\star$ can be obtained through an \textit{affine transformation} of $f$, which cannot increase the degree, and the same transformation can similarly be applied to any approximating polynomial of $f$ (and the converse follows since $(f^\star)^\star = f$). The second component of the approximation scheme is the following lemma.

\begin{lemma}[\cite{beniamini2020approximate}, similar to \cite{sherstov2020algorithmic}]
	\label{lemma:low_l1_apx_deg}
	Let $f: \{0,1\}^n \rightarrow \{0,1\}$ be a Boolean function, and let $p \in \mathbb{R}[x_1, \dots, x_n]$ be its representing polynomial, where $\norm{p}_1 \in [3, 2^n]$. Then:
	$$
	\forall \norm{p}_1^{-1} \le \varepsilon \le \frac13 :\ \ \epsadeg{f} = \mathcal{O}\left(\sqrt{n \log \norm{p}_1} \right)
	$$
\end{lemma}

The proof of Lemma \ref{lemma:low_l1_apx_deg} follows from the following simple approximation scheme: replace every \textit{monomial} (of sufficiently large degree) with a polynomial that approximates it pointwise, to some sufficiently small error (depending only on the $\ell_1$-norm of the representing polynomial). The full details of this scheme appeared previously in \cite{beniamini2020approximate, sherstov2020algorithmic}. Combining Lemma \ref{lemma:low_l1_apx_deg} with the $\ell_1$-bound of Corollary \ref{cor:l1norm_ubpmnstar}, we obtain:

\begin{corollary}
	\label{cor:ubpm_apx_deg}
	For any $n > 1$, and $2^{-n \log n} \le \varepsilon \le \frac13$, we have:
	$$ \epsadeg{\UBPMn} = \mathcal{O}(n^{3/2} \sqrt{\log n}) $$
\end{corollary}

\section{Families of Matching Functions having Low Dual $\ell_1$-Norm}
\label{appendix:families_of_matching_functions}

The main algorithmic result in this paper is the low $\ell_1$-norm of the dual function of $\UBPMn$, from which we deduce \textit{upper bounds}, for instance on the communication rank and the approximate degree. In \cite{beniamini2020approximate}, a similar bound had been obtained for the dual of the perfect matching function, $\BPMn$. These norm bounds and their corollaries extend to a wide range of \textit{matching-related functions}, some of which are detailed below. 

\paragraph*{Functions Obtained by Restrictions.} Consider any two Boolean functions $f$ and $g$, such that $g$ is obtained by a \textit{restriction} of $f$ (i.e., by fixing some of the inputs bits of $f$). As  restrictions cannot increase the norm, it clearly holds that $\norm{g}_1 \le \norm{f}_1$ and $\norm{g^\star}_1 \le \norm{f^\star}_1$. Several intrinsically interesting matching-functions can be cast in this way. One notable example is the bipartite $k$-matching function, which is the indicator over all graphs $G \subseteq K_{n,n}$ containing a matching of size $k$.
$$ \BMnk(x_{1,1}, \dots, x_{n,n}) = \begin{cases} 1 & \big\{ (i,j) : x_{i,j} = 1\big\} \subseteq K_{n,n} \text{ has a $k$-matching } \\ 0 & \text{otherwise} \end{cases} $$
This function is obtained by a restriction of $\BPM_{2n-k}$, as follows. Label the vertices of $K_{2n-k,2n-k}$ by
\begin{equation*}
\begin{split}
	\quad\quad\quad\quad\quad\quad L &= A \sqcup V, \text{ where } A=\{a_1, \dots, a_n\}, V=\{v_1, \dots, v_{n-k}\} \\
	R &= B \sqcup U, \text{ where } B=\{b_1, \dots, b_n\}, U=\{u_1, \dots, u_{n-k}\}	
\end{split}
\end{equation*}
Given any input $G \subseteq K_{n,n}$ to $\BMnk$, the edges of $G$ are encoded via the edges joining $A$ and $B$, and moreover we fix two additional bicliques $K_{A,U}$, $K_{V,B}$. The resulting graph contains a bipartite perfect matching if and only if $G$ contains a $k$-matching, and thus
\begin{corollary*}
	\label{cor:norm_bmnk}
	For every $0 < k \le n$, we have $\log \norm{\BMnkstar}_1 = \mathcal{O}(n \log n)$.
\end{corollary*}
This norm bound is tight whenever $k = \alpha n$, for any \textit{constant} $0 < \alpha < 1$, as are (up to log-factors) the bounds on the approximate degree and on the log-rank.

\begin{corollary*}
	\label{cor:bm_bounds}
	Let $\alpha \in (0,1)$ be a constant. Then for every $n > 1$ and $2^{-n \log n} \le \varepsilon \le \frac13$, we have:
	$$ \log \norm{\BM^\star_{n, \alpha n}}_1 = \Theta(n \log n),\ \ \epsadeg{\BM_{n, \alpha n}} = \widetilde{\Theta}(n^{\sfrac{3}{2}}), \text{ and } \log \rank \left(\BM_{n, \alpha n} \right) = \widetilde{\Theta}(n) $$
\end{corollary*}

The aforementioned approximate degree lower bound follows using the method of \textit{Spectral Sensitivity} -- a complexity measure due to Aaronson, Ben-David, Kothari, Rao and Tal \cite{aaronson2020degree}, based on Huang's proof of the sensitivity conjecture \cite{huang2019induced}. \cite{aaronson2020degree} proved that the approximate degree of any total function $f$ is bounded below by the spectral radius of its \textit{sensitivity graph} (i.e., the $f$-cut of the hypercube). As this graph is bipartite, its spectrum is symmetric, and it therefore suffices (by Cauchy interlacing) to obtain a lower bound on the spectral radius of any vertex induced subgraph of the sensitivity graph \cite{beniamini2020approximate}.

For $\BMnk$ this construction is straightforward -- consider the induced graph whose left vertices are all $(k-1)$-matchings, and right vertices are all $k$-matchings. This produces a biregular subgraph of the sensitivity graphs of $\BMnk$, with left degrees $(n-k+1)^2$ and right degrees $k$. As it is well known that the spectral radius of a biregular graph is $\sqrt{d_L d_R}$ (where $d_L$ and $d_R$ are the left and right degrees, respectively), this concludes the bound on the spectral sensitivity of $\BMnk$, and by extension, its approximate degree\footnote{We remark that the same construction also trivially extends to $\UBMnk$; the \textit{unique} $k$-matching function.}. This lower bound on $\adeg{\BMnk}$ now implies the $\ell_1$-norm lower bound, through Lemma \ref{lemma:low_l1_apx_deg}.

As for the log-rank lower bound, it follows by a simple fooling set argument, under the same input partition used in Theorem \ref{thm:rank_ubpm}. Let $L = A \sqcup B$ be the left vertices corresponding to the input partition, where $|A| = |B| = \sfrac{n}{2}$, and let $A'$ and $B'$ be the first $\sfrac{k}{2}$ vertices of $A$ and $B$, respectively. Let $C$ be the first $k = \alpha n$ right vertices. Then, under the notation of Theorem~\ref{thm:rank_ubpm},
$$ S = \Big\{ \left(\id(A',S), \id(B',\bar{S})\right)\ :\ S \subseteq C,\ \bar{S} = C \setminus S,\ |S|=|\bar{S}|=\tfrac{k}2 \Big\} $$
is a fooling set for $\BM_{n, \alpha n}$, where the indices of $S$ and $\bar{S}$ correspond to a \textit{fixed} ordering on $C$. Any pair $(x,y)$ contains a $k$-matching, but for any mismatching pair belonging to sets $S_1 \ne S_2 \subseteq C$, we have that $S_1 \cap S_2 \ne \emptyset$ and thus the maximum matching is of size $|S_1 \cup S_2| < k$. By construction, this fooling set is of size
$$ \log_2 |S| = \log_2 {k \choose {\sfrac{k}{2}}} = k - o(1) $$
and the log-rank bound now follows from Theorem \ref{thm:fooling_set_bounds_log_rank}.\footnote{For the \textit{unique} bipartite $k$-matching function $\UBMnk$ one can obtain a slightly stronger log-rank bound by repeating the construction of Theorem \ref{thm:rank_ubpm} with $k$-matchings rather than perfect matchings, and by adding $n-k$ isolated vertices. This yields a log-rank bound of $\log_2 \left(\sfrac{k}{2} !\right) = \Theta(k \log k)$.}

\paragraph*{Formulas over Low-Norm Functions} Given any two nontrivial Boolean functions $f$ and $g$, the norms of their conjunction, disjunction, and negation are at-most multiplicative in their respective norms, and the same holds for their duals. Therefore, the dual of any short De Morgan formula whose atoms are Boolean functions of low dual $\ell_1$-norm, will similarly inherit the low-norm property. Several matching functions can be represented in this way, and thus have low dual norm. For example
$$ \MaxMatchnk(x_{1,1}, \dots, x_{n,n}) = \begin{cases} 1 & \text{The maximum matching of } \big\{ (i,j) : x_{i,j} = 1\big\} \text{ is of size $k$} \\ 0 & \text{otherwise} \end{cases} $$
can be constructed as $\MaxMatchnk = \BMnk \land \lnot \BM_{n, k+1}$.

\end{document}